\documentclass[11pt, a4paper]{scrartcl}
\usepackage[british]{babel}
\usepackage[utf8]{inputenc}

\usepackage{amsmath}
\usepackage{amssymb}
\usepackage{mathtools}
\usepackage{physics}

\usepackage{amsthm}
\theoremstyle{definition}
\newtheorem{definition}{Definition}[section]

\theoremstyle{theorem}
\newtheorem{lemma}[definition]{Lemma}
\newtheorem{theorem}[definition]{Theorem}

\usepackage{tikz}
\usetikzlibrary{matrix,arrows.meta}
\usepackage{tikz-cd}

\usepackage{url}
\usepackage{hyperref}

\numberwithin{equation}{section}

\title{A group-theoretic characterisation
of Taub-Nut spacetime}
\author{Schiden Yohannes and Domenico Giulini\\
Institute for Theoretical Physics\\
Leibniz University of Hannover\\
Appelstrasse 2, 30167 Hannover\\ 
Germany}
\begin{document}
\date{\today}
\maketitle
\begin{abstract} 
\noindent
We prove that any $G=SU(2)\times U(1)$ symmetric spacetime 
that is Ricci flat (i.e. solves the matter-free $\Lambda=0$ 
Einstein equations) with non-null $G$-orbits is locally 
isometric to some maximally extended generalised Taub-NUT spacetime. 
\end{abstract}
\section{Introduction}
This paper deals with Taub-NUT spacetime 
\cite{Taub:1951,Newman.Tamburino.Unti:1963} 
that is an exact solution to the matter-free 
(i.e. vacuum) equations of General Relativity 
without cosmological constant. This particular 
solution is known for its many surprising 
properties 
\cite{Misner:1963,Misner:1967,Misner.Taub:1969}.
Also, it can be generalised in various ways 
and embedded into the wider Pleba{\'n}ski–Demia{\'n}ski
class of solutions; see, e.g., 
\cite[chapter\,12,16]{Griffiths.Podolsky:2009}.
Here we are not interested in these generalisations 
and restrict attention to strict Taub-NUT only.  

We prove a result that gives rise to a new 
group-theoretic characterisation of Taub-NUT spacetime, 
or rather some obvious topological generalisation of it. 
The main theorem will be presented and proved in 
Section\,\ref{sec:MainTheorem}. It states that 
\emph{any} \(SU(2)\times U(1)\) symmetric vacuum solution 
to Einstein's equations with non-null group orbits is 
locally isometric to \emph{some} maximally extended  
generalised Taub-NUT geometry, where the ``generalisation'' 
here consists in replacing the 3-sphere in the global 
\(\mathbb{R}\times S^3\) topology with that of a lens 
space \(L(n,1)\). In that sense we can now say that 
(generalised) Taub-NUT can be characterised by its 
isometry group, together with a mild restriction on 
the group orbits. 

This result may be seen as an (incomplete) analogue 
to the so-called Jebsen-Birkhoff theorem, going back 
in idea to \cite{Jebsen:1921,Jebsen:2006} and 
(without much proof) \cite{Birkhoff:RMP1923}; see \cite{Johansen.Ravndal:2006} for more
on its history and multiple discovery. 
In a modern formulation it states that any 
spherically symmetric vacuum solution to Einstein's
equation is locally isometric to the maximally 
extended Schwarzschild-Kruskal manifold. Here
``spherical symmetry'' is defined by the 
existence of an isometric action of $SO(3)$ with 
spacelike \(S^2$ or $\mathbb{R}P^2\) orbits. 
A modern proof can be found in 
\cite[Chapter 4.10.1-2]{Straumann:GR2013}.

We call our result an \emph{incomplete} 
analogue to the Jebsen-Birkhoff theorem because
in the generalised Taub-NUT case we have several 
inequivalent (i.e. non globally isometric) maximal 
extensions, the precise classification of which 
we currently investigate.

There is another, different  notion of 
\emph{generalised Taub-NUT} introduced in 
\cite{Moncrief:1984}, which relaxes the global 
isometry to be merely \(U(1)\) (i.e. dropping 
the \(SU(2)\) factor altogether) and requires 
the spacetime to contain a compact Cauchy horizon 
diffeomorphic to \(S^3\) to which the \(U(1)\) 
action restricts to a free action with lightlike 
orbits. Hence the \(S^3\) Cauchy horizon is the 
total space of a \(U(1)\) principal bundle with 
base \(S^2\) (Hopf bundle) and lightlike fibres. 
The set of such ``generalised Taub-NUT'' 
spacetimes forms an infinite-dimensional 
proper submanifold within the set of all 
\(U(1)\)-symmetric vacuum spacetimes of ``roughly
half the dimensions'' \cite[p.\,108]{Moncrief:1984}.
In this case there are uncountably many 
inequivalent maximal extensions. 

In order to make our paper self contained, we will 
review the essential geometry and topology of 
Taub-NUT spacetime in Section\,\ref{ref:Taub-NUT},
also providing a characterisation of the 
NUT-charge as ``dual'' to mass. 
Our main theorem concerning the group-theoretic 
characterisation will be stated and proved in 
Section\,\ref{sec:MainTheorem}. We end with a 
brief outlook in Section\,\ref{sec:Outlook}.

\section{Taub-NUT space-time}
\label{ref:Taub-NUT}
In the following we will present some of the 
features of the Taub-NUT space-time with 
respect to the interpretation given by 
Misner~\cite{Misner:1963}. In this interpretation 
the topology of space-time is \(\mathbb{R} \times S^3\) 
and, using Euler coordinates, the metric is given by
\begin{align}
	\label{mismetric}
	g = -4l^2f(r)(d\psi+\cos\theta d\varphi)^2+\frac{1}{f(r)}dr^2+(r^2+l^2)(d\theta^2+\sin^2\theta d\varphi^2),
\end{align}
with
\begin{align}
	f(r) = \frac{r^2-2mr-l^2}{r^2+l^2}.
\end{align}
The constant \(m \in \mathbb{R}\) is interpreted as 
the mass and the constant 
\(l \in \mathbb{R}\setminus\{0\}\) is referred 
to as the NUT parameter. The four-dimensional 
isometry group of this space-time is 
\(SU(2)_L \times U(1)_R\) induced by the 
left-invariant vector field \(\xi_0 = \partial_\psi\) 
(generating right-translations) 
and the right-invariant vector fields
\begin{subequations}
	\begin{align}
		\xi_1 &= -\sin\varphi\partial_\theta- 				\cot\theta\cos\varphi\partial_\varphi+\csc\theta\cos\varphi\partial_\psi\\
		\xi_2 &= \cos\varphi\partial_\theta- 	\cot\theta\sin\varphi\partial_\varphi+\csc\theta\sin\varphi\partial_\psi\\
		\xi_3 &= \partial_\varphi.
	\end{align}
\end{subequations}
(generating left-translations) on the 3-sphere,
which we identify with the group manifold of 
\(SU(2)\). Note that the subscripts \(L\) and 
\(R\) on \(SU(2)\) and \(U(1)\), respectively,  
are meant to indicate that these groups act 
via left- and right-multiplication on  \(SU(2)\).
The vector fields \(\xi_0,\xi_1,\xi_2,\xi_3\) 
satisfy the commutation relations:
\begin{subequations}
	\begin{align}
		[\xi_i,\xi_j] &= -\varepsilon^k_{ij}\xi_k \\
		[\xi_0,\xi_i] &= 0 \quad \qquad i,j,k = 1,2,3.
	\end{align}
\end{subequations}
In terms of the left-invariant one-forms
\begin{subequations}
	\begin{align}
		\sigma_x &= \sin\psi d\theta-\sin\theta\cos\psi d\varphi\\
		\sigma_y &= \cos\psi d\theta+\sin\theta\sin\psi d\varphi\\
		\sigma_z &= d\psi+\cos\theta d\varphi
	\end{align}
\end{subequations}
the metric can be written as
\begin{align}
	g = -4l^2f(r)\sigma_z^2+\frac{1}{f(r)}dr^2+(r^2+l^2)(\sigma_x^2+\sigma_y^2).
\end{align} 
The orbit generated by \(\xi_1,\xi_2,\xi_3\) is three-dimensional, namely \(S^3\), with the orbits generated by \(\xi_0\) being subsets of it.\\

In the given coordinates the analytical expressions  
become singular at \(r_{\pm} = m \pm \sqrt{m^2+l^2}\),
whereas all components of the Riemann tensor in an
orthonormal tetrad, and hence in particular the Kretschmann
 scalar, are regular. This indicates that these singularities are, in fact, coordinate artefacts.
They correspond to the Killing horizons of the Killing vector field \(\partial_\psi\). A possible coordinate transformation removing these singularities is given by
\begin{align}
	\psi' = \psi + \int\frac{1}{2lf(r)}dr
\end{align}
such that the metric in these new coordinates is given by
\begin{align}
	\label{firstextension}
	g = -4l^2f(r)(d\psi'+\cos\theta d\varphi)^2+2(2l)(d\psi'+\cos\theta d\varphi)dr+(r^2+l^2)(d\theta^2+\sin^2\theta d\varphi^2).
\end{align}
Another coordinate transformation would be
\begin{align}
	\psi'' = \psi - \int\frac{1}{2lf(r)}dr,
\end{align}
giving
\begin{align}
	g = -4l^2f(r)(d\psi''+\cos\theta d\varphi)^2-2(2l)(d\psi''+\cos\theta d\varphi)dr+(r^2+l^2)(d\theta^2+\sin^2\theta d\varphi^2).
\end{align}
Written in terms of the left-invariant one-forms of \(S^3\), it is immediate that the metrics are regular on the whole manifold \(\mathbb{R} \times S^3\). Furthermore, it can be shown that both space-times are maximal \cite{Taub:1951}. In these coordinates both the stationary regions \(r<r_-\) and \(r>r_+ \, \), the so-called NUT-regions, and the region \(r_-<r<r_+\), called the Taub-region, are included. In particular, the hypersurfaces of \(r = const.\) are 3-spheres being spacelike in the Taub-region, timelike in the NUT-regions, and lightlike at \(r=r_\pm\).\\

Furthermore, with respect to the \(U(1)\) right 
multiplication the space-time can be considered to be a principal fibre bundle analogous to the Hopf bundle. 
Since the \(r = const.\) hypersurfaces are 3-spheres, 
there exist no equal-time hypersurfaces intersecting 
these 3-spheres in two-spheres along which we could 
evaluate the Komar integral for mass in the usual 
form . However, we can use the
structure of \(S^3\) as \(U(1)\) principle
fibre bundle over the base \(S^2\), which has a 
natural connection  given by the distribution of 
orthogonal complements to the fibre in each 
tangent space where the generating vector 
field of \(U(1)\) is non-null.  It is then  
possible to uniquely identify horizontal 
and right-invariant \(k\)-forms with \(k\)-forms on the base \(S^2\). Thus, considering the NUT regions, admitting the timelike Killing vector field \(\partial_\psi\) 
(generating the right-\(U(1)\) translation) 
the Komar mass of the space-time can be calculated. 
We will be using the orthonormal tetrad
\begin{subequations}
	\begin{align}
		\vartheta^0 &= 2lf^{1/2}(r)(d\psi+\cos\theta d\varphi)\\
		\vartheta^1 &= f^{-1/2}(r)dr\\
		\vartheta^2 &= (r^2+l^2)^{1/2} d\theta\\
		\vartheta^3 &= (r^2+l^2)^{1/2}\sin\theta d\varphi.	
	\end{align}
\end{subequations}
For \(\lim\limits_{r \to \infty} f(r) = 1\), we will calculate the Komar mass with respect to the Killing vector field \(k \coloneqq -\frac{1}{2l}\partial_\psi\) which is 
normalised at infinity \(r \to \infty\). The metric-dual one-form of the timelike Killing vector field is then given by
\begin{align}
	k^{\flat} = 2lf(r)(d\psi+\cos\theta d\varphi)
\end{align}
and hence
\begin{subequations}
	\begin{align}
		dk^{\flat} &= 2lf'(r)dr \wedge (d\psi+\cos\theta d\varphi) - 2lf(r) \sin\theta d\theta \wedge d\varphi\\
		&= -f'(r) \vartheta^0 \wedge \vartheta^1 -2l\frac{f(r)}{r^2+l^2} \vartheta^2 \wedge \vartheta^3\,,
	\end{align}
\end{subequations}
where we used the standard notation that 
denotes the one-form image of the vector 
\(k\) under the metric isomorphism by 
\(k^{\flat}=:g(k,\cdot)\). Now, picking  
the orientation defined by \(\omega = \vartheta^0 \wedge \vartheta^1 \wedge \vartheta^2 \wedge \vartheta^3\), we get
\begin{subequations}
	\begin{align}
		*dk^{\flat} &= f'(r) \vartheta^2 \wedge \vartheta^3 - 2l\frac{f(r)}{r^2+l^2} \vartheta^0 \wedge \vartheta^1\\
		&= f'(r)(r^2+l^2) \sin\theta d\theta \wedge d\varphi + 4l^2\frac{f(r)}{r^2+l^2} dr \wedge (d\psi+\cos\theta d\varphi).
	\end{align}
\end{subequations}
Then with respect to an arbitrary hypersurface \(r_0 = const.\), the two-form is given by
\begin{align}
	*dk^{\flat} &= -f'(r_0)(r_0^2+l^2) d\sigma_z.
\end{align}
Since \(\frac{i}{2}\sigma_z\) is a connection one-form for the principal fibre bundle, \(*dk^{\flat}\) can be considered as a multiple of the curvature form and hence is horizontal and right-invariant, as well as closed. Thus, we can identify it with a closed two-form on the base space \(S^2\), such that using the formula for the Komar mass, we have
\begin{align}
	-\frac{1}{8\pi}\int_{S^2_{\infty}}*dk^{\flat} = \lim_{r\to\infty}-\frac{1}{8\pi}f'(r)(r^2+l^2)\int_{S^2}\sin\theta d\theta \wedge d\varphi = -m.
\end{align}
Therefore \(m\) can be interpreted as the Komar mass of the space-time. Moreover, considering \(dk^{\flat}\) instead of \(*dk^{\flat}\), the same line of argument can be applied to give
\begin{align}
	-\frac{1}{8\pi}\int_{S^2_{\infty}}dk^{\flat} =  \lim_{r\to\infty}\frac{1}{8\pi}2lf(r)\int_{S^2}\sin\theta d\theta \wedge d\varphi = l.
\end{align}
So the constants \(m\) and \(l\) are 
related by Hodge duality.\\

Duality also arises in the description of the dual-Bondi-mass of space-times which are asymptotically empty and flat at null infinity and with vanishing Bondi news. It can be shown that in this case null infinity, for space-times having a non-vanishing dual-Bondi-mass, is topologically a Lens space \(L(n,1)\) and a principal fibre bundle \((L(n,1),\pi,S^2;S^1)\), with the dual-Bondi-mass being proportional to the number of twists, \(n\), in the bundle. Conversely, if null infinity is a non-trivial \(S^1\) principal fibre bundle over \(S^2\), the news tensor field vanishes and there exists an infinitesimal translation such that the dual-Bondi-mass with respect to it is non-zero. In particular, the Taub-NUT space-time can be shown to be asymptotically empty and flat at null infinity, with null infinity being a 3-sphere. The dual-Bondi-mass with respect to the infinitesimal translation induced by the Killing vector field \(-\frac{1}{2l}\partial_\psi\) can be computed to be the NUT parameter \(l\)~\cite{Ramaswamy.Sen:1981}.

\section{Main theorem}
\label{sec:MainTheorem}
In this section we intend to give a unique characterisation of the Taub-NUT space-time in terms of the isometry group and its orbits. In particular, the Taub-NUT space-time can be seen to be the universal cover of a family of space-times admitting \(SU(2) \times U(1)\) as an isometry group such that the group orbits of \(SU(2)\times U(1)\) and \(SU(2)\) are three-dimensional and non-null.\\

Since the metric of the Taub-NUT space-time induces a \(SU(2)_L \times U(1)_R\) invariant metric on the hypersurfaces \(r = const.\), being diffeomorphic to \(SU(2)\), we will begin by studying special metrics on \(SU(2)\). For Lorentz metrics on \(SU(2)\) we have

\begin{lemma}\label{u1timelike}
	Let \(G = SU(2)\) and \(g\) a Lorentz metric on \(G\) such that it is \(SU(2)\) left-invariant and \(U(1)\) right-invariant, whereby \(U(1)_R\) is considered as a subgroup of the \(SU(2)_R\). Then the orbits of the \(U(1)\) right-multiplication are timelike curves.
\end{lemma}

\begin{proof}
	The left-action\footnote{The reader 
should be aware of the conceptual difference between left/right-multiplication on groups
and left/right-action of groups on sets: 
A \emph{left-action} of a group $G$ on a 
set $S$ is simply a homomorphism 
$\Phi:G\rightarrow\mathrm{Bij}(S)$, 
from the group $G$ into the group of 
bijections of $S$, with group multiplication 
of the latter just being composition of maps. 
This means that $\Phi:g\mapsto\Phi_g$
is such that $\Phi_g\circ\Phi_h=\Phi_{gh}$.
In contrast, a \emph{right-action} is an 
anti-homomorphism 
$\tilde\Phi:G\rightarrow\mathrm{Bij}(S)$
which satisfies $\tilde\Phi_g\circ\tilde\Phi_h=\tilde\Phi_{hg}$.
A left-action can be turned into a 
right-action (and vice versa) if we compose it 
with the group inversion $I:G\rightarrow G$, 
$g\mapsto I(g):=g^{-1}$ which is an 
anti-homomorphims, i.e. $I(gh)=I(h)I(g)$.
Then $\tilde\Phi:=\Phi\circ I$ is a 
right-action if $\Phi$ is a left-action. Now,     
If $S=G$, there are two natural actions 
of $G$ on itself, called $L$ and $R$ and 
given by left- and right-multiplication
respectively: $L_g(p):=gp$ and $R_g(p):=pg$.
Associativity of group multiplication implies
that these two actions commute (as maps):
$L_g\circ R_h=R_h\circ L_g$ for all 
$g,h\in G$. Written in this way $L$ is a 
left- and $R$ is a right-action which 
as such do not combine to any action, left or 
right. However, the 
right-multiplication can be turned into a 
left-action by composing $R$ with $I$. 
In this way we get two different and 
commuting left-actions of $G$ on itself, 
one by left-multiplication with $g\in G$ and one by right-multiplication with $g^{-1}$. 
Together they define a left-action of 
$G\times G$ on $G$, given by 
$\Phi_{(g,h)}:=L_g\circ R_{h^{-1}}$,
that is $\Phi_{(g,h)}(p)=gph^{-1}$. In order to 
distinguish the group $G$ that acts by 
left-multiplication from the one that acts 
by right-multiplication (with the inverse) we 
distinguish them notationally and call them $G_L$ and $G_R$, respectively. Restricting 
this to the diagonal subgroup 
$G_\Delta:=\{(g,g):G\in G\}\subset G_L\times G_R$ 
gives the left-action of $G$ on itself that 
is usually referred to as ``conjugation''.} \(SU(2)_L\times U(1)_R\subset S(U)_L\times SU(2)_R\) is simply obtained by restricting the standard 
left-action of 
\(S(U)_L\times SU(2)_R\) on $SU(2)$: 
	\begin{subequations}
		\begin{align}
			&(SU(2)_L \times U(1)_R) \times SU(2) \to SU(2)\\
			&((h,h'),g) \mapsto hgh'^{-1}.
		\end{align}
	\end{subequations}
	Now let \(e \in G\) be the identity, then we have for \(h \in SU(2)_L\) and \(h' \in U(1)_R\)
	\begin{align}
		((h,h'),e) \mapsto heh'^{-1} = hh'^{-1}.
	\end{align}
	Hence the isotropy group at \(e\) is the diagonal \(U(1)\) subgroup in \(SU(2)_L \times U(1)_R\), 
denoted by  \(C_h\). Then
	\begin{align}
		(dC_h)_e : T_eG \to T_eG
	\end{align}
	induces the adjoint representation
	\begin{subequations}
		\begin{align}
			Ad: U(1) \to GL(T_eG)\\
			Ad(h) = (dC_h)_e.
		\end{align}
	\end{subequations}
	Since the tangent space is three-dimensional, the action induced by this \(U(1)\) on the tangent space is given by a \(U(1)\) subgroup of the three-dimensional Lorentz group. Furthermore, because the \(U(1)\) subgroups of the three-dimensional Lorentz group consist of rotations acting by orthogonal transformations in a spacelike plane, such that the corresponding orthogonal timelike direction is invariant, the three-dimensional tangent space decomposes into an orthogonal sum of a two-dimensional spacelike subspace and a one-dimensional timelike subspace. Choosing any normed vector \(v\) in the one-dimensional timelike subspace, we define the left-invariant vector field \(X \in Lie(G)\) by
	\begin{align}
		X(g) = \left. \dv{}{t}\right|_{\scriptscriptstyle t=0}(g \, exp(tv)).
	\end{align}
	The right action on this left-invariant vector field is determined by the adjoint of \(v \in T_eG\) with respect to \(h^{-1} \in U(1)\),
	\begin{subequations}
		\begin{align}
			(dR_h)_g(X_g) &= \left. \dv{}{t}\right|_{\scriptscriptstyle t=0}(g \, exp(tv) \, h)\\
			&= \left. \dv{}{t}\right|_{\scriptscriptstyle t=0}(gh \, C_{h^{-1}}(exp(tv)))\\
			&= \left. \dv{}{t}\right|_{\scriptscriptstyle t=0}(gh \,exp(tAd(h^{-1})(v))).
		\end{align}
	\end{subequations}
	Since the timelike direction is invariant with respect to the adjoint representation, we obtain
	\begin{subequations}
		\begin{align}
			(dR_h)_g(X_g)&= \left. \dv{}{t}\right|_{\scriptscriptstyle t=0}(gh \,exp(tAd(h^{-1})(v)))\\
			&= \left. \dv{}{t}\right|_{\scriptscriptstyle t=0}(gh \, exp(tv))\\
			&= X(gh).
		\end{align}
	\end{subequations}
	Therefore the left-invariant vector field is also 
\(U(1)_R\)--invariant. Furthermore, because it is a left-invariant vector field, it generates a \(U(1)_R\)-action, considered as \(U(1)_R' \subset SU(2)_R\). By being also invariant under \(U(1)_R \subset SU(2)_R\), the two \(U(1)\) right actions have to commute, hence \(U(1)_R = U(1)_R'\). Therefore, the orbits of the \(U(1)\) right action coincide with the orbits of \(X\) and thus are timelike.
\end{proof}

Next we will prove that a \(SU(2)_L \times U(1)_R\)--invariant metric on \(SU(2)\) can be put into a canonical form:

\begin{lemma}\label{canonmetric}
	Let \(G = SU(2)\) and \(g\) a non-degenerate, symmetric bilinear form on \(G\) which is \(SU(2)_L \times U(1)_R\)--invariant. Then \(g\) can be written as
	\begin{align}
		g = A \sigma_z^2+B(\sigma_x^2+\sigma_y^2),
	\end{align}
	where \(\sigma_x,\sigma_y,\sigma_z\) are left--invariant one-forms on \(G\).
\end{lemma}

\begin{proof}
	Let \(Z\) be a fundamental vector field associated to the \(U(1)_R\)-action and any element \(ix \in Lie(U(1)) =i\mathbb{R}\),
	\begin{align}
		Z(g) = \left. \dv{}{t}\right|_{\scriptscriptstyle t=0}(g \, exp(tix)), \quad g \in G.
	\end{align}
	Then the vector field \(Z\) is left-invariant. We will complete it to a basis for \(Lie(G)\) by choosing two linearly independent left-invariant vector fields in the orthogonal complement of \(Z\), so \(X,Y \in Lie(G)\) such that \(X,Y \bot Z\). Then, denoting the basis as \(e_1 = Z, e_2 = X, e_3 = Y\) and their dual one forms by \(\omega_1,\omega_2,\omega_3\), \(g\) can be written as
	\begin{align}
		g = \lambda (\omega^1)^2+ \mu (\omega^2)^2+ \nu (\omega^3)^2+ \kappa \omega^2\omega^3.
	\end{align}
	Now since \(g\) is \(U(1)_R\)--invariant, we have
	\begin{align}
		L_{e_1}g = 0.
	\end{align}
	If the structure constants are given by
	\begin{align}
		[e_i,e_j] = c^k_{ij} e_k
	\end{align}
	we have for their dual one-forms
	\begin{align}
		d\omega^k = - \sum_{i<j} c^k_{ij} \omega^i \wedge \omega ^j.
	\end{align}
	Thus, using Cartan's magic formula, we obtain
	\begin{subequations}
		\begin{align}
			L_{e_1}\omega^1 &= i_{e_1}d\sigma^1 + d(i_{e_1}\sigma^1)\\
			&= -i_{e_1}\left(\sum_{i<j} c^1_{ij} \omega^i \wedge \omega ^j\right)\\
			&= -\sum_{i<j} c^1_{ij} \left( (i_{e_1}\omega^i) \wedge \omega^j - \omega^i \wedge (i_{e_1}\omega^j) \right)\\
			&= -\left(c^1_{12} \, \omega^2 + c^1_{13} \, \omega^3\right)
		\end{align}
	\end{subequations}
	and similarly
	\begin{subequations}
		\begin{align}
			L_{e_1}\omega^2 &= -\left(c^2_{12} \, \omega^2 + c^2_{13} \, \omega^3\right)\\
			L_{e_1}\omega^3 &= -\left(c^3_{12} \, \omega^2 + c^3_{13} \, \omega^3\right).
		\end{align}
	\end{subequations}
	Because
	\begin{align}
		L_V(T \otimes S) = (L_VT) \otimes S + T \otimes (L_VS),
	\end{align}
	for any vector field \(V\) and tensor fields \(T,S\), we see that
	\begin{align}
		L_{e_1}g = 0 \implies c^1_{12} = c^1_{13} = 0,
	\end{align}
	by noting that terms of the form \(\omega^1\omega^2\) and \(\omega^1\omega^3\) can only be obtained by \(L_{e_1}\omega^1\). Thus we have
	\begin{subequations}
		\begin{align}
			&[e_1,e_2] \in span\{e_2,e_3\}\\
			&[e_1,e_3] \in span\{e_2,e_3\}.
		\end{align}
	\end{subequations}
	Defining a vector space endomorphism \(F:Lie(SU(2)) \to Lie(SU(2))\) by
	\begin{align}
		F(e_1) = [e_2,e_3] \quad F(e_2) = [e_3,e_1] \quad F(e_3) = [e_1,e_2],
	\end{align}
	the matrix of \(F\) with respect to \(\{e_1,e_2,e_3\}\) is given by
	\begin{align}
		F = \begin{pmatrix*}
			c^1_{23} & 0 & 0\\
			c^2_{23} & c^2_{31} & c^2_{12}\\
			c^3_{23} & c^3_{31} & c^3_{12}
		\end{pmatrix*}.
	\end{align}
	Since \(SU(2)\) is a unimodular group, we have: \(tr \, ad(x) = 0 \quad \forall x \in Lie(SU(2))\), therefore implying \(c^2_{23} = c^3_{23} = 0\). Since we also have \(c^2_{12} = c^3_{31}\), \(F\) is self-adjoint with respect to \(g\) and hence we can find an orthogonal transformation diagonalizing \(F\). It is an orthogonal transformation keeping \(e_1\) fixed and transforming in its orthogonal complement. Thus, by the appropriate transformation, we get a new basis \(\{\tilde{e}_1,\tilde{e}_2,\tilde{e}_3\}\) with \(\tilde{e}_1 = e_1\) satisfying
	\begin{subequations}
		\begin{align}
			&[\tilde{e}_1,\tilde{e}_2] = \tilde{c}^3_{12} \tilde{e}_3\\
			&[\tilde{e}_3,\tilde{e}_1] = \tilde{c}^2_{31} \tilde{e}_2\\
			&[\tilde{e}_2,\tilde{e}_3] = \tilde{c}^1_{23} \tilde{e}_1
		\end{align}
	\end{subequations}
	and by rescaling, we obtain the basis \(\{e_1',e_2',e_3'\}\) with the commutation relations
	\begin{subequations}
		\begin{align}
			&[e_1',e_2'] = e_3'\\
			&[e_3',e_1'] = e_2'\\
			&[e_2',e_3'] = e_1'.
		\end{align}
	\end{subequations}
	Now, denoting the dual one-forms of this basis by \(\{\sigma_z,\sigma_x,\sigma_y\}\) respectively, \(g\) can be written as
	\begin{align}
		g = A \sigma_z^2+B \sigma_x^2+ C \sigma_y^2 + D \sigma_x \sigma_y
	\end{align}
	and we obtain
	\begin{subequations}
		\begin{align}
			L_{e_1'}\sigma_z &= 0\\
			L_{e_1'}\sigma_x &= \sigma_y\\
			L_{e_1'}\sigma_y &= -\sigma_x.
		\end{align}
	\end{subequations}
	Then, since \(e_1'\) is just a scalar multiple of \(e_1\), the \(U(1)\) right invariance implies
	\begin{subequations}
		\begin{align}
			0 &= L_{e_1'}g\\
			&= B(\sigma_y \otimes \sigma_x +\sigma_x \otimes \sigma_y) + C(-\sigma_x \otimes \sigma_y -\sigma_y \otimes \sigma_x)+D(\sigma_y \otimes \sigma_y - \sigma_x \otimes \sigma_x)\\
			&= (B-C)(\sigma_y \otimes \sigma_x + \sigma_x \otimes \sigma_y) + D(\sigma_y \otimes \sigma_y - \sigma_x \otimes \sigma_x).
		\end{align}
	\end{subequations}
	Thus, implying \(B = C, \, D = 0\), such that
	\begin{align}
		g = A \sigma_z^2+B (\sigma_x^2+\sigma_y^2)
	\end{align}
\end{proof}

Combining the results of Lemma \ref{u1timelike} and \ref{canonmetric}, we see that a \(SU(2)_L \times U(1)_R\)--invariant Riemannian/Lorentz metric \(g\) on \(SU(2)\) can always be written as
\begin{align}
	g = \varepsilon A^2 \sigma_z^2+B^2 (\sigma_x^2+\sigma_y^2),
\end{align}
where \(\varepsilon = 1\) corresponds to the Riemannian and 
\(\varepsilon = -1\) to the Lorentzian case. Using Euler-angle coordinates the metric is given by
\begin{align}
	g = \varepsilon A^2 (d\psi+\cos\theta d\varphi)^2 + B^2 (d\theta^2+\sin^2\theta d\varphi^2).
\end{align}

An essential observation regarding the orbits of the isometry group \(SU(2) \times U(1)\) of the Taub-NUT space-time is that the orbit corresponding to the \(SU(2)\)-action are three-dimensional and the orbit with respect to the \(U(1)\)-action is a subset of it. Thus simply requiring the group orbits of a general space-time with isometry group \(SU(2) \times U(1)\) to be three-dimensional, does not exclude the possibility that the action of \(SU(2)\) generates two-dimensional orbits and the action of \(U(1)\) transversal one-dimensional orbits. For that matter we will in the following study the implications of three-dimensional orbits generated by a \(SU(2)\)-action on a space-time.\\

Let \(M\) be a manifold admitting a \(SU(2)\) left-action in such a way that the group orbits are three-dimensional. Since \(SU(2)\) is a compact Lie group the action is proper and thus each group orbit is a closed subset of \(M\) and each isotropy group is compact. Let \(O(p), \, I_p\) denote the orbit and isotropy group of \(p\) respectively. Then since \(I_p\) is a closed subgroup and the action of \(SU(2)\) on its group orbits is transitive, we have
\begin{align}
	dim \, O(p) = dim \, SU(2) - dim \, I_p.
\end{align}
Thus the isotropy group is a discrete subgroup of \(SU(2)\) and since it is compact it has to be finite. The orbit \(O(p)\) is a homogeneous \(SU(2)\)-space and we have
\begin{align}
	O(p) \cong SU(2)/I_p.
\end{align}
Since for a connected Lie group \(G\) and a discrete subgroup \(\Gamma\) the quotient is a manifold and the quotient map is a (normal) covering map, we get a fibration with base space  \(O(p) \cong SU(2)/I_p\), discrete fibers \(I_p\) and total space \(SU(2)\)
\begin{center}
	\begin{tikzcd}
	I_p\arrow[r]
	& SU(2)\arrow[d] \\
	& O(p) \cong SU(2)/I_p.
	\end{tikzcd}
\end{center}
For \(SU(2)\) is connected and simply connected it is the universal cover. By the use of the long exact sequence of homotopy groups for the fibration we obtain
\begin{align}
	\cdots \rightarrow \pi_1(I_p) \rightarrow \pi_1(SU(2)) \rightarrow \pi_1(O(p)) \rightarrow \pi_0(I_p) \rightarrow \pi_0(SU(2)).
\end{align}
Noting that \(SU(2) \cong S^3\), \(\pi_1(SU(2))\) is trivial like \(\pi_0(SU(2))\). Furthermore, we have \(\pi_0(I_p) \cong I_p\) and thus obtain the short exact sequence
\begin{align}
	0 \rightarrow \pi_1(O(p)) \rightarrow I_p \rightarrow 0,
\end{align}
implying \(\pi_1(O(p)) \cong I_p\). Summarizing, we see that the group orbits are closed three-dimensional manifolds with finite fundamental group. But then by Thurston's elliptisation conjecture
(now proven) the group orbits have to be elliptic 3-manifolds. These have been classified to be of the form \(M = S^3/\Gamma\), with \(\pi_1(M) = \Gamma\) being a finite subgroup of \(SO(4)\), acting freely and orthogonally on \(M\)
in the standard fashion. Out of these, the only ones admitting \(SU(2) \times U(1)\) as an isometry group are the Lens spaces \(L(n,1)\) \cite{Hong.EtAl:2012}.
\\

As already proven a \(SU(2)_L \times U(1)_R\)--invariant metric on \(SU(2) \cong S^3\) can be put into a canonical form. Now we want to study the case for \(L(n,1)\). Considering \(S^3 \subset \mathbb{C}^2\) the left action of \(SU(2)\) on \(S^3 \) is the natural action of \(SU(2)\) on \(\mathbb{C}^2\) and the \(\Gamma = \mathbb{Z}_n\)-action on \(S^3\) for \(L(n,1)\) is given by
\begin{align}
	\label{zn-action}
	(z_0,z_1) \mapsto (e^{2\pi i/n}z_0,e^{2\pi i/n}z_1), \quad (z_0,z_1) \in S^3. 
\end{align}
Now we can define a left action of \(SU(2)\) on \(L(n,1)\) by
\begin{subequations}
	\begin{align}
		&SU(2) \times L(n,1) \to L(n,1)\\
		&(A, \pi(p)) \mapsto \pi(Ap),
	\end{align}
\end{subequations}
where \(\pi: S^3 \to L(n,1)\) is the projection map, so the covering map. This induces a well-defined \(SU(2)\) left action on \(L(n,1)\). Similarly we have a well-defined induced \(U(1)\) right action.\\

Moreover, given a \(SU(2)_L \times U(1)_R\)--invariant metric it is also invariant with respect to the \(\mathbb{Z}_n\)-action \eqref{zn-action} and hence the following construction is well-defined:\\

Let \(g\) be a \(\mathbb{Z}_n\)--invariant metric on \(S^3\). We define a metric on \(L(n,1)\) pointwise by
\begin{subequations}
	\begin{align}
		&g'_q :T_qL(n,1) \times T_qL(n,1) \to \mathbb{R} \\
		&g'_q (X',Y') \coloneqq g_p(X,Y), \quad X',Y' \in T_qL(n,1), \, X,Y \in T_pS^3
	\end{align}
\end{subequations}
where \(\pi(p) = q\) and \(d\pi_p(X) = X', \, d\pi_p(Y) = Y'\). The \(\mathbb{Z}_n\)-invariance of \(g\) implies the independence of the choice of a representative \(p\) and since the covering map is a local diffeomorphism its differential at any point is a linear isomorphism. Therefore this definition makes sense. The metric \(g'\) is in fact smooth, since given any smooth local section \(s:U\subset L(n,1) \to S^3\) and smooth local vector fields \(X',Y'\) on \(U\) we have on \(U\)
\begin{align}
	g'(X',Y') = g(ds(X'),ds(Y')).
\end{align}
This linear map is a bijection, because conversely given any metric \(g'\) on \(L(n,1)\) we can define a metric \(g = \pi^*g'\) on \(S^3\), being just the preimage of the preceding construction. Thus, we see that the \(\mathbb{Z}_n\)--invariant metrics on \(S^3\) are in bijection to metrics on \(L(n,1)\). Since a \(SU(2)_L \times U(1)_R\)--invariant metric on \(S^3\) is also invariant with respect to the \(\mathbb{Z}_n\)-action \eqref{zn-action}, by definition of the induced action and the constructed bijection above, it is immediate that the \(SU(2)_L \times U(1)_R\)--invariant metrics on \(S^3\) are in one-to-one correspondence to the metrics on \(L(n,1)\) which are invariant under the induced \(SU(2)_L \times U(1)_R\)-action.\\

Using Euler coordinates and Lemma \ref{canonmetric}, a Lorentz or Riemannian metric on \(L(n,1)\) invariant under \(SU(2)_L \times U(1)_R\) can always be written as
\begin{align}
	g = \varepsilon A^2(d\psi+\cos\theta d\varphi)^2 + B^2 (d\theta^2+\sin^2\theta d\varphi^2), \quad \varepsilon = \pm 1
\end{align}
where \(\theta\) and \(\varphi\) ranging from \(0\) to \(\pi\) and \(0\) to \(2\pi\) respectively and \(\psi\) being \(4\pi/n\)-periodic.\\

In particular, based on the preceding results, we see that the metric \eqref{mismetric} is invariant with respect to the \(\mathbb{Z}_n\)-action \eqref{zn-action} and hence there exists a well-defined metric on any Lens space \(L(n,1)\). Thus the space-time can be generalized to \((\mathbb{R}\times L(n,1),g)\) with 
\begin{align}
	g = -4l^2f(r)(d\psi+\cos\theta d\varphi)^2+\frac{1}{f(r)}dr^2+(r^2+l^2)(d\theta^2+\sin^2\theta d\varphi^2),
\end{align}
which will be called the \emph{generalized Taub-NUT space-time}. The space-time \((\mathbb{R}\times L(n,1),g')\), with
\begin{align}
	g' =  -4l^2f(r)(d\psi'+\cos\theta d\varphi)^2+2(2l)(d\psi'+\cos\theta d\varphi)dr+(r^2+l^2)(d\theta^2+\sin^2\theta d\varphi^2)
\end{align}
will be called a \emph{maximal extension of the generalized Taub-NUT space-time}.\\

\noindent
Now we can prove the following statement:

\begin{theorem}
	Every (\(C^2\)-) solution to the vacuum Einstein field equations admitting \(SU(2) \times U(1)\) as an isometry group, such that \(SU(2)\times U(1)\) and \(SU(2)\) both have three-dimensional non-null orbits in an open subset \(U\), is locally isometric to a maximal extension of the generalized Taub-NUT space-time.
\end{theorem}

\begin{proof}
	Let \(p \in M\) be an arbitrary point of the space-time \((M,g)\) and \(O(p)\) be the three-dimensional orbit of \(p\) with respect to the action of \(SU(2)\). We define the orthogonal complement of the tangent space of the orbit to be: \(N_p \coloneqq T_pO(p)^{\perp}\). Then the induced distributions
	\begin{align}
		N \coloneqq \cup_{p \in M}N_p, \quad O \coloneqq \cup_{p \in M}T_pO(p)
	\end{align}
	are both integrable, for \(N\) is a one-dimensional distribution, which is always integrable, and \(O\) is by construction integrable, with its integral manifolds being the orbits. So we have a involutive three-dimensional distribution, spanned by Killing vector fields and the involutive one-dimensional normal bundle with \(N \cap O = {0}\), since the orbits are non-null. Thus, it is possible to introduce local coordinates \(\{x^\mu\} = \{r,x^1,x^2,x^3\}\) such that
	\begin{align}
		g = g_{rr} dr^2 + g_{ab}(x^\mu)dx^adx^b
	\end{align}
	where \(r = const.\) are the integral manifolds of \(O\), the orbits, which are homogeneous spaces. Since the three-dimensional orbits of a \(SU(2)\)-action admitting \(SU(2) \times U(1)\) as isometry group have to be topologically the Lens spaces \(L(n,1)\) and the \(SU(2)_L \times U(1)_R\)--invariant Lorentz or Riemannian metrics on \(L(n,1)\) can always be put into a canonical form, the metric can be written as
	\begin{align}
		g = -\varepsilon A^2(r) dr^2 + \varepsilon B^2(r) (d\psi+\cos\theta d\varphi)^2 + R^2(r) (d\theta^2+\sin^2\theta d\varphi^2),
	\end{align}
	where the case \(\varepsilon = 1\) represents spacelike orbits and  \(\varepsilon = -1\) timelike orbits.\\
	
	\noindent
	Now to solve the field equations, it is necessary to calculate the corresponding Ricci tensor. We will calculate them using an orthonormal tetrad and the Cartan structure equations.\\
	
	\noindent
	First we will consider the case \(\varepsilon = 1\), so spacelike orbits. The orthonormal tetrad we will be using is
	\begin{subequations}
		\begin{align}
			\vartheta^0 &= A(r) dr\\
			\vartheta^1 &= B(r)(d\psi+\cos\theta d\varphi)\\
			\vartheta^2 &= R(r) d\theta\\
			\vartheta^3 &= R(r)\sin\theta d\varphi.
		\end{align}
	\end{subequations}
	In the following the argument of the functions \(A,B,R\) will be omitted and a prime indicates the derivative with respect to \(r\). Then exterior differentiation and expressing the results in terms of the tetrad leads to
	\begin{subequations}
		\begin{align}
			d\vartheta^0 &= 0\\
			d\vartheta^1 &= \frac{B'}{AB} \vartheta^0 \wedge \vartheta^1 - \frac{B}{R^2} \vartheta^2 \wedge \vartheta^3\\
			d\vartheta^2 &= \frac{R'}{AR} \vartheta^0 \wedge \vartheta^2\\
			d\vartheta^3 &=\frac{R'}{AR} \vartheta^0 \wedge \vartheta^3 + \frac{\cot\theta}{R} \vartheta^2 \wedge \vartheta^3.
		\end{align}
	\end{subequations}
	Since the tetrad is orthonormal we have \(\omega_{\mu\nu} + \omega_{\nu\mu} = 0\). Now using the first structure equation with an ansatz for every connection one-form of the form \(a_\mu\vartheta^\mu\), the unique solution is given by
	\begin{subequations}
		\begin{align}
			\omega^0{}_{1} &= \omega^1{}_{0} = \frac{B'}{AB} \vartheta^1\\
			\omega^0{}_{2} &= \omega^2{}_{0} = \frac{R'}{AR} \vartheta^2\\
			\omega^0{}_{3} &= \omega^3{}_{0} = \frac{R'}{AR} \vartheta^3\\
			\omega^1{}_{2} &= -\omega^2{}_{1} = - \frac{B}{2R^2} \vartheta^3\\
			\omega^1{}_{3} &= -\omega^3{}_{1} = \frac{B}{2R^2} \vartheta^2\\
			\omega^2{}_{3} &= -\omega^3{}_{2} = \frac{B}{2R^2} \vartheta^1 -  \frac{\cot\theta}{R} \vartheta^3.
		\end{align}
	\end{subequations}

Then using the second structure equations, the curvature \(2\)-form can be calculated:
	\begin{subequations}
		\begin{align}
			\Omega^0{}_1 &= d\omega^0{}_1+\omega^0{}_2 \wedge \omega^2{}_1 + \omega^0{}_3 \wedge \omega^3{}_1 \nonumber\\
			&= \left(\frac{B''}{AB}-\frac{B'A'}{A^2B}-\frac{B'^2}{AB^2}\right)dr \wedge \vartheta^1 + \frac{B'}{AB} d\vartheta^1 + \frac{BR'}{2AR^3} \vartheta^2 \wedge \vartheta^3 - \frac{BR'}{2AR^3} \vartheta^3 \wedge \vartheta^2 \nonumber\\
			&= \left(\frac{B''}{A^2B}-\frac{B'A'}{A^3B}-\frac{B'^2}{A^2B^2}\right) \vartheta^0 \wedge \vartheta^1 + \frac{B'^2}{A^2B^2} \vartheta^0 \wedge \vartheta^1  - \frac{B'}{AR^2} \vartheta^2 \wedge \vartheta^3 + \frac{BR'}{AR^3} \vartheta^2 \wedge \vartheta^3 \nonumber\\
			&= \left(\frac{B''}{A^2B}-\frac{B'A'}{A^3B}\right) \vartheta^0 \wedge \vartheta^1 + \left(\frac{BR'}{AR^3} - \frac{B'}{AR^2}\right) \vartheta^2 \wedge \vartheta^3\\
			\Omega^0{}_2 &= d\omega^0{}_2+\omega^0{}_1 \wedge \omega^1{}_2 + \omega^0{}_3 \wedge \omega^3{}_2 \nonumber\\
			&= \left(\frac{R''}{A^2R}-\frac{R'A'}{A^3R}-\frac{R'^2}{A^2R^2}\right) \vartheta^0 \wedge \vartheta^2+ \frac{R'^2}{A^2R^2} \vartheta^0 \wedge \vartheta^2 - \frac{B'}{2AR^2} \vartheta^1 \wedge \vartheta^3 - \frac{BR'}{2AR^3} \vartheta^3 \wedge \vartheta^1 \nonumber\\
			&= \left(\frac{R''}{A^2R}-\frac{R'A'}{A^3R}\right) \vartheta^0 \wedge \vartheta^2 		-\frac{1}{2}\left(\frac{B'}{AR^2} - \frac{BR'}{AR^3}\right) \vartheta^1 \wedge \vartheta^3\\
			\Omega^0{}_3 &= d\omega^0{}_3+\omega^0{}_1 \wedge \omega^1{}_3 + \omega^0{}_2 \wedge \omega^2{}_3 \nonumber\\
			&= \left(\frac{R''}{A^2R}-\frac{R'A'}{A^3R}-\frac{R'^2}{A^2R^2}\right) \vartheta^0 \wedge \vartheta^3+ \frac{R'^2}{A^2R^2} \vartheta^0 \wedge \vartheta^3 + \frac{R'\cot\theta}{AR^2} \vartheta^2 \wedge \vartheta^3 + \frac{B'}{2AR^2} \vartheta^1 \wedge \vartheta^2 \nonumber \\& \hspace{1em}+ \frac{BR'}{2AR^3} \vartheta^2 \wedge \vartheta^1- \frac{R'\cot\theta}{AR^2} \vartheta^2 \wedge \vartheta^3 \nonumber \\
			&= \left(\frac{R''}{A^2R}-\frac{R'A'}{A^3R}\right) \vartheta^0 \wedge \vartheta^3 +\frac{1}{2}\left(\frac{B'}{AR^2} - \frac{BR'}{AR^3}\right) \vartheta^1 \wedge \vartheta^2
		\end{align}
	\end{subequations}

	\begin{subequations}
		\begin{align}
			\Omega^1{}_2 &= d\omega^1{}_2+\omega^1{}_0 \wedge \omega^0{}_2 + \omega^1{}_3 \wedge \omega^3{}_2 \nonumber\\
			&= \left(-\frac{B'}{2AR^2} + \frac{BR'}{AR^3}\right) \vartheta^0 \wedge \vartheta^3 - \frac{BR'}{2AR^3} \vartheta^0 \wedge \vartheta^3 - \frac{B\cot\theta}{2R^3}\vartheta^2 \wedge \vartheta^3 + \frac{B'R'}{A^2BR} \vartheta^1 \wedge \vartheta^2 \nonumber \\& \hspace{1em} - \frac{B^2}{4R^4} \vartheta^2 \wedge \vartheta^1 + \frac{B\cot\theta}{2R^3}\vartheta^2 \wedge \vartheta^3 \nonumber \\
			&= \frac{1}{2}\left(\frac{BR'}{AR^3} - \frac{B'}{AR^2} \right) \vartheta^0 \wedge \vartheta^3 + \left( \frac{B'R'}{A^2BR} + \frac{B^2}{4R^4}\right) \vartheta^1 \wedge \vartheta^2\\
			\Omega^1{}_3 &= d\omega^1{}_3+\omega^1{}_0 \wedge \omega^0{}_3 + \omega^1{}_2 \wedge \omega^2{}_3 \nonumber\\
			&= \left(\frac{B'}{2AR^2} - \frac{BR'}{AR^3}\right) \vartheta^0 \wedge \vartheta^2 + \frac{BR'}{2AR^3} \vartheta^0 \wedge \vartheta^2 + \frac{B'R'}{A^2BR} \vartheta^1 \wedge \vartheta^3 - \frac{B^2}{4R^4} \vartheta^3 \wedge \vartheta^1 \nonumber\\
			&= \frac{1}{2}\left(\frac{B'}{AR^2} - \frac{BR'}{AR^3}\right) \vartheta^0 \wedge \vartheta^2 + \left( \frac{B'R'}{A^2BR} + \frac{B^2}{4R^4}\right) \vartheta^1 \wedge \vartheta^3\\
			\Omega^2{}_3 &= d\omega^2{}_3+\omega^2{}_0 \wedge \omega^0{}_3 + \omega^2{}_1 \wedge \omega^1{}_3 \nonumber\\
			&= \left(\frac{B'}{2AR^2} - \frac{BR'}{AR^3}\right) \vartheta^0 \wedge \vartheta^1 + \frac{B'}{2AR^2} \vartheta^0 \wedge \vartheta^1 - \frac{B^2}{2R^4} \vartheta^2 \wedge \vartheta^3 + \frac{\csc^2\theta}{R^2} \vartheta^2 \wedge \vartheta^3 \nonumber\\ &\hspace{1em}+\frac{R'\cot\theta}{AR^2} \vartheta^0 \wedge \vartheta^3 - \frac{R'\cot\theta}{AR^2} \vartheta^0 \wedge \vartheta^3 - \frac{\cot^2\theta}{R^2} \vartheta^2 \wedge \vartheta^3 + \frac{R'^2}{A^2R^2} \vartheta^2 \wedge \vartheta^3 +\frac{B^2}{4R^4} \vartheta^3 \wedge \vartheta^2 \nonumber\\
			&= \left(\frac{B'}{AR^2} - \frac{BR'}{AR^3}\right) \vartheta^0 \wedge \vartheta^1 + \left(\frac{1}{R^2}+ \frac{R'^2}{A^2R^2}-\frac{3}{4}\frac{B^2}{R^4}\right) \vartheta^2 \wedge \vartheta^3.
		\end{align}
	\end{subequations}
	Using \(\Omega^\mu{}_\nu = \frac{1}{2}R^{\mu}_{\nu\alpha\beta}\theta^\alpha\wedge\theta^\beta\), the non-vanishing components of the Riemann tensor, up to symmetry, are
	\begin{subequations}
		\begin{align}
			R^0{}_{101} &= \frac{B''}{A^2B}-\frac{B'A'}{A^3B}\\
			R^0{}_{123} &= 2R^0{}_{213} = -2R^0{}_{312} = \frac{BR'}{AR^3} - \frac{B'}{AR^2}\\
			R^0{}_{202} &= R^0{}_{303} = \frac{R''}{A^2R}-\frac{R'A'}{A^3R}\\
			R^1{}_{212} &= R^1{}_{313} = \frac{B'R'}{A^2BR} + \frac{B^2}{4R^4}\\
			R^2{}_{323} &= \frac{1}{R^2}+ \frac{R'^2}{A^2R^2}-\frac{3}{4}\frac{B^2}{R^4}
		\end{align}
	\end{subequations}
	The components of the Ricci tensor are
	\begin{subequations}
		\begin{align}
			R_{00} &= -R^0{}_{101} - 2 R^0{}_{202}\\
			R_{11} &= R^0{}_{101} + 2 R^1{}_{212}\\
			R_{22} &= R_{33} = R^0{}_{202} + R^1{}_{212} + R^2{}_{323},
			\intertext{with all other components being zero. The Ricci scalar is then}
			R &= -R_{00} + R_{11} + R_{22} + R_{33}\\
			&= 2\left(R^0{}_{101}+2R^1{}_{212}+2R^0{}_{202}+ R^2{}_{323}\right).
			\intertext{Thus, we obtain the following non-vanishing components of the Einstein tensor}
			G_{00} &= 2 R^1{}_{212} + R^2{}_{323}\\
			G_{11} &= -2R^0{}_{202} - R^2{}_{323}\\
			G_{22} &= G_{33} = -R^0{}_{202} - R^1{}_{212} - R^0{}_{101}.
		\end{align}
	\end{subequations}
	Now we will solve the vacuum Einstein field equations. Suppose \(\langle dR,dR\rangle = 0\) in \(U\). Then, if \(R\) is constant in \(U\), \(R = R_0 >0 \), we have
	\begin{subequations}
		\begin{align}
			0 &= G_{00} = \frac{B^2}{2R^4_0}+\frac{1}{R^2_0}-\frac{3}{4}\frac{B^2}{R^4_0}\\
			\iff B^2 &= 4R^2_0 = \text{const}.
		\end{align}
	\end{subequations}
	Inserting this condition in \(0 = G_{11}\) we arrive at the following 
contradiction: 
	\begin{align}
		0 = G_{11} &= -\frac{1}{R^2_0}+\frac{3}{4}\frac{B^2}{R^4_0} = \frac{2}{R^2_0}.
	\end{align}
On the other hand, \(dR\) can not be non-zero and lightlike, since the group orbits are non-null. Therefore, considering \(\langle dR,dR\rangle \neq 0\) in \(U\), we will choose coordinates such that \(R(r) = r\). Next, to solve the field equations, we will consider
	\begin{subequations}
		\begin{align}
			0 &= G_{00} + G_{11}= 2\left(R^1{}_{212}-R^0{}_{202}\right)\\
			\iff 0 &= r^3(AB'+A'B)+\frac{1}{4}A^3B^3\\
			&= r^3(AB)'+\frac{1}{4}(AB)^3.
		\end{align}
	\end{subequations}
	Defining \(D(r) = A(r)B(r)\), we get the first order ordinary differential equation for \(D\)
	\begin{align}
		r^3D'+\frac{1}{4}D^3 = 0,
	\end{align}
	which can be integrated to give
	\begin{align}
		D^2(r) = \frac{4r^2c_0}{r^2-c_0},
	\end{align}
	with \(c_0 > 0\). Thus we need to have \(r^2 > c_0\). Now to solve
	\begin{align}
		0 &= G_{00} = 2\frac{B'}{A^2Br} - \frac{B^2}{4r^4} + \frac{1}{r^2} + \frac{1}{A^2r^2}
	\end{align}
	we use \(D^2 = A^2B^2\) and multiply the equation by \(4c_0r^4\) to get
	\begin{subequations}
		\begin{align}
			0 &= r(r^2-c_0)2BB'- c_0B^2 + 4r^2c_0 + (r^2-c_0)B^2\\
			&= r(r^2-c_0)(B^2)'+(r^2-2c_0)B^2 + 4r^2c_0.
		\end{align}
	\end{subequations}
	Then, by introducing \(F(r) \coloneqq B^2(r)\), we obtain an inhomogeneous first order linear differential equation
	\begin{align}
		0 = F'+\frac{r^2-2c_0}{r(r^2-c_0)}F + \frac{4rc_0}{r^2-c_0}.
	\end{align}
	To solve this differential equation, we will first solve the corresponding homogeneous equation:
	\begin{subequations}
		\begin{align}
			0 &= F_h'+\frac{r^2-2c_0}{r(r^2-c_0)}F_h\\
			&= F_h'+\frac{2(r^2-c_0)-r^2}{r(r^2-c_0)}F_h\\
			&= F_h'+\left(\frac{2}{r}-\frac{r}{r^2-c_0}\right)F_h.
		\end{align}
	\end{subequations}
	Thus, by integrating, the solution is given by
	\begin{subequations}
		\begin{align}
			&ln(F_h) = -\left(ln(r^2)-ln\left(\sqrt{r^2-c_0}\right)\right)+c\\
			\iff & F_h = c'\frac{\sqrt{r^2-c_0}}{r^2}, \quad c' = e^c.
		\end{align}
	\end{subequations}
	Now to obtain the general solution, we will multiply the inhomogeneous equation by \(c'/F_h\):
	\begin{subequations}
		\begin{align}
			0 &= \frac{r^2}{\sqrt{r^2-c_0}}F'+\frac{r^3-2rc_0}{(r^2-c_0)^{3/2}}F + \frac{4r^3c_0}{(r^2-c_0)^{3/2}}\\
			&= \left(\frac{r^2}{\sqrt{r^2-c_0}}F\right)' + \frac{4r^3c_0}{(r^2-c_0)^{3/2}}.
		\end{align}
	\end{subequations}
	Hence, the solution is given by
	\begin{subequations}
		\begin{align}
			F = c_1\frac{\sqrt{r^2-c_0}}{r^2}-\frac{\sqrt{r^2-c_0}}{r^2}\int\frac{4r^3c_0}{(r^2-c_0)^{3/2}}dr.
		\end{align}
	\end{subequations}
	To solve the integral we will use the substitution \(u = \sqrt{r^2-c_0}\) with \(dr = \frac{u}{\sqrt{u^2+c_0}}du\) such that
	\begin{subequations}
		\begin{align}
			&\int \frac{4(u^2+c_0)^{3/2}c_0}{u^3}\frac{u}{\sqrt{u^2+c_0}}du\\
			= &\int 4c_0+\frac{4c_0{}^2}{u^2}du\\
			= &4c_0u-\frac{4c_0{}^2}{u}.
		\end{align}
	\end{subequations}
	Therefore, the solution of the inhomogeneous differential equation is
	\begin{align}
		F = B^2 = \frac{-4c_0r^2+c_1\sqrt{r^2-c_0}+8c_0^2}{r^2}.
	\end{align}
	Thus, we have solved \(G_{00} = G_{11} = 0\) and obtained functions \(B^2(r)\) and \(A^2(r) = \frac{D^2(r)}{B^2(r)}\) implying
	\begin{subequations}
		\begin{align}
			R^1{}_{212} &=  R^0{}_{202}\\
			2R^1{}_{212} &= - R^2{}_{323}.
		\end{align}
	\end{subequations}
	One can check that these functions satisfy \(R^0{}_{101} = R^2{}_{323}\), implying \(G_{22} = G_{33} = 0\).\\
	
	Inserting the functions in the metric, we see that there is a singularity at \(r^2 = c_0\). To check if the space-time has a curvature singularity, we will compute the Kretschmann scalar
	\begin{align}
	K = R^{\alpha\beta\mu\nu}R_{\alpha\beta\mu\nu}.
	\end{align}
	Using the symmetry of the Riemann tensor and the implications of the field equations, the Kretschmann scalar is given by
	\begin{align}
		K &=  4(R^0{}_{101})^2-8(R^0{}_{123})^2-8(R^0{}_{231})^2-8(R^0{}_{231})^2+4(R^0{}_{202})^2 \nonumber\\
		&\hspace{1em}+4(R^0{}_{303})^2+4(R^1{}_{212})^2+4(R^1{}_{313})^2+4(R^2{}_{323})^2 \nonumber\\
		&= 12\left((R^0{}_{101})^2-(R^0{}_{123})^2\right) \nonumber\\
		&= \frac{3}{4} c_0{}^{-2}r^{-12}\biggl(2048c_0{}^6-3072c_0{}^5r^2-18c_0c_1{}^2r^4+c_1{}^2r^6+128c_0{}^4\left(9r^4+4c_1\sqrt{r^2-c_0}\right)\\& \quad \qquad \qquad \quad  +48c_0{}^2c_1r^2\left(c_1+2r^2\sqrt{r^2-c_0}\right)-32c_0{}^3\left(c_1{}^2+2r^6+16c_1r^2\sqrt{r^2-c_0}\right)\biggr),
	\end{align}
	observing that it is regular at \(r^2 = c_0\),  thus indicating that the singularity is due to a poor choice of coordinates. To solve the coordinate singularity, we will use the following coordinate transformation:
	\begin{align}
		r' = \frac{1}{2\sqrt{c_0}}\int D(r)dr = \sqrt{r^2-c_0}, \quad dr' = \frac{1}{2\sqrt{c_0}}D(r)dr.
	\end{align}
	With respect to \(r'\) we have
	\begin{align}
		B^2(r') = -4c_0\frac{r'^2-\frac{c_1}{4c_0}r'-c_0}{r'^2+c_0}.
	\end{align}
	Defining \(l^2 \coloneqq c_0 > 0, \, m \coloneqq \frac{c_1}{8c_0}\), the metric takes the form
	\begin{subequations}
		\begin{align}
			g &= -\frac{4l^2}{B^2(r')}dr'^2+B^2(r')(d\psi+\cos\theta d\varphi)^2 + (r'^2+l^2) \, (d\theta^2+\sin^2\theta d\varphi^2)\\
			&= \frac{r'^2+l^2}{r'^2-2mr-l^2}dr'^2 - 4l^2\frac{r'^2-2mr'-l^2}{r'^2+l^2} (d\psi+\cos\theta d\varphi)^2 + (r'^2+l^2) \, (d\theta^2+\sin^2\theta d\varphi^2).
		\end{align}
	\end{subequations}
	So we see that we obtain the generalized Taub-NUT space-time and since we assumed that the orbits are spacelike, we get in fact the Taub-region. Now if we consider the case \(\varepsilon = -1\), so timelike orbits, we have
	\begin{align}
		g =  A^2(r) dr^2 - B^2(r)  (d\psi+\cos\theta d\varphi)^2 + R^2(r) (d\theta^2+\sin^2\theta d\varphi^2).
	\end{align}
	Using the orthonormal tetrad
	\begin{subequations}
		\begin{align}
			\vartheta^0 &= B(r)(d\psi+\cos\theta d\varphi)\\
			\vartheta^1 &= A(r) dr\\
			\vartheta^2 &= R(r) d\theta\\
			\vartheta^3 &= R(r)\sin\theta d\varphi.
		\end{align}
	\end{subequations}
	we obtain
	\begin{subequations}
		\begin{align}
			d\vartheta^0 &= \frac{B'}{AB} \vartheta^1 \wedge \vartheta^0 - \frac{B}{R^2} \vartheta^2 \wedge \vartheta^3\\
			d\vartheta^1 &= 0\\
			d\vartheta^2 &= \frac{R'}{AR} \vartheta^1 \wedge \vartheta^2\\
			d\vartheta^3 &=\frac{R'}{AR} \vartheta^1 \wedge \vartheta^3 + \frac{\cot\theta}{R} \vartheta^2 \wedge \vartheta^3.
		\end{align}
	\end{subequations}
	The connection one-forms are then given by
	\begin{subequations}
		\begin{align}
			\omega^0{}_{1} &= \omega^1{}_{0} = \frac{B'}{AB} \vartheta^0\\
			\omega^0{}_{2} &= \omega^2{}_{0} = -\frac{B}{2R^2} \vartheta^3\\
			\omega^0{}_{3} &= \omega^3{}_{0} = \frac{B}{2R^2}\vartheta^2\\
			\omega^1{}_{2} &= -\omega^2{}_{1} = - \frac{R'}{AR} \vartheta^2\\
			\omega^1{}_{3} &= -\omega^3{}_{1} = -\frac{R'}{AR} \vartheta^3\\
			\omega^2{}_{3} &= -\omega^3{}_{2} = -\frac{B}{2R^2} \vartheta^0 -  \frac{\cot\theta}{R} \vartheta^3.
		\end{align}
	\end{subequations}
	Due to their strong resemblance to the case of spacelike orbits we will see that solving the vacuum Einstein equations is analogous. Using the second structure equations the curvature \(2\)-form can be calculated
	\begin{subequations}
		\begin{align}
			\Omega^0{}_1 &= -\left(\frac{B''}{A^2B}-\frac{B'A'}{A^3B}\right) \vartheta^0 \wedge \vartheta^1 + \left(\frac{BR'}{AR^3} - \frac{B'}{AR^2}\right) \vartheta^2 \wedge \vartheta^3\\
			\Omega^0{}_2 &= \frac{1}{2}\left(\frac{BR'}{AR^3} - \frac{B'}{AR^2} \right) \vartheta^1 \wedge \vartheta^3 - \left( \frac{B'R'}{A^2BR} + \frac{B^2}{4R^4}\right) \vartheta^0 \wedge \vartheta^2\\
			\Omega^0{}_3 &= \frac{1}{2}\left(\frac{B'}{AR^2} - \frac{BR'}{AR^3}\right) \vartheta^1 \wedge \vartheta^2 - \left( \frac{B'R'}{A^2BR} + \frac{B^2}{4R^4}\right) \vartheta^0 \wedge \vartheta^3\\
			\Omega^1{}_2 &= -\left(\frac{R''}{A^2R}-\frac{R'A'}{A^3R}\right) \vartheta^1 \wedge \vartheta^2 -\frac{1}{2}\left(\frac{B'}{AR^2} - \frac{BR'}{AR^3}\right) \vartheta^0 \wedge \vartheta^3\\
			\Omega^1{}_3 &= -\left(\frac{R''}{A^2R}-\frac{R'A'}{A^3R}\right) \vartheta^1 \wedge \vartheta^3 -\frac{1}{2}\left(\frac{B'}{AR^2} - \frac{BR'}{AR^3}\right) \vartheta^0 \wedge \vartheta^2\\
			\Omega^2{}_3 &=\left(\frac{B'}{AR^2} - \frac{BR'}{AR^3}\right) \vartheta^0 \wedge \vartheta^1 + \left(\frac{1}{R^2}- \frac{R'^2}{A^2R^2}+\frac{3}{4}\frac{B^2}{R^4}\right) \vartheta^2 \wedge \vartheta^3.
		\end{align}
	\end{subequations}
	Then using \(\Omega^\mu{}_\nu = \frac{1}{2}R^{\mu}_{\nu\alpha\beta}\theta^\alpha\wedge\theta^\beta\), the non-vanishing components of the Riemann tensor, up to symmetry, are
	\begin{subequations}
		\begin{align}
			R^0{}_{101} &= -\frac{B''}{A^2B}+\frac{B'A'}{A^3B}\\
			R^0{}_{123} &= 2R^0{}_{213} = -2R^0{}_{312} = \frac{BR'}{AR^3} - \frac{B'}{AR^2}\\
			R^0{}_{202} &= R^0{}_{303} = -\frac{B'R'}{A^2BR} - \frac{B^2}{4R^4}\\
			R^1{}_{212} &= R^1{}_{313} = -\frac{R''}{A^2R}+\frac{R'A'}{A^3R}\\
			R^2{}_{323} &= \frac{1}{R^2}- \frac{R'^2}{A^2R^2}+\frac{3}{4}\frac{B^2}{R^4}.
		\end{align}
	\end{subequations}
	Analogously, we have the following non-vanishing components of the Einstein tensor
	\begin{subequations}
		\begin{align}
			G_{00} &= 2 R^1{}_{212} + R^2{}_{323}\\
			G_{11} &= -2R^0{}_{202} - R^2{}_{323}\\
			G_{22} &= G_{33} = -R^0{}_{202} - R^1{}_{212} - R^0{}_{101}.
		\end{align}
	\end{subequations}
	Now we will solve the vacuum Einstein field equations. Supposing \(\langle dR,dR\rangle = 0\) in \(U\) leads analogously to a contradiction. Therefore, we have \(\langle dR,dR\rangle \neq 0\) in \(U\) and we can choose coordinates such that \(R(r) = r\). Then to solve the field equations we will consider again \(0 = G_{00} + G_{11} =  2\left(R^1{}_{212}-R^0{}_{202}\right)\), but since \(R^1{}_{212}\) and \(R^0{}_{202}\) correspond to curvature components \(-R^0{}_{202}\) and \(-R^1{}_{212}\), respectively, for the case \(\varepsilon = 1\) we see that they satisfy the same differential equation. Hence defining \(D(r) = A(r)B(r)\), we have
	\begin{align}
		D^2(r) = \frac{4r^2c_0}{r^2-c_0},
	\end{align}
	with \(c_0 > 0\) and \(r^2 > c_0\). Now we will solve
	\begin{align}
		0 &= G_{11} = 2\frac{B'}{A^2Br} - \frac{B^2}{4r^4} - \frac{1}{r^2} + \frac{1}{A^2r^2}.
	\end{align}
	Using \(D^2 = A^2B^2\) and multiplying the equation by \(4c_0r^4\) we obtain
	\begin{subequations}
		\begin{align}
			0 &= r(r^2-c_0)2BB'- c_0B^2 - 4r^2c_0 + (r^2-c_0)B^2\\
			&= r(r^2-c_0)(B^2)'+(r^2-2c_0)B^2 - 4r^2c_0\,.
		\end{align}
	\end{subequations}
	Then again introducing \(F(r) \coloneqq B^2(r)\) we have the inhomogeneous first order linear differential equation
	\begin{align}
		0 = F'+\frac{r^2-2c_0}{r(r^2-c_0)}F - \frac{4rc_0}{r^2-c_0}.
	\end{align}
	Hence, we see that the corresponding homogeneous equation coincides with the one for spacelike orbits. Thus, we have
	\begin{align}
		F_h = c'\frac{\sqrt{r^2-c_0}}{r^2}.
	\end{align}
	Now to obtain the general solution we will multiply the inhomogeneous equation by \(c'/F_h\):
	\begin{subequations}
		\begin{align}
			0 &= \frac{r^2}{\sqrt{r^2-c_0}}F'+\frac{r^3-2rc_0}{(r^2-c_0)^{3/2}}F - \frac{4r^3c_0}{(r^2-c_0)^{3/2}}\\
			&= \left(\frac{r^2}{\sqrt{r^2-c_0}}F\right)' - \frac{4r^3c_0}{(r^2-c_0)^{3/2}}.
		\end{align}
	\end{subequations}
	Hence, the solution is given by
	\begin{subequations}
		\begin{align}
			F = B^2 &= c_1\frac{\sqrt{r^2-c_0}}{r^2}+\frac{\sqrt{r^2-c_0}}{r^2}\int\frac{4r^3c_0}{(r^2-c_0)^{3/2}}dr\\
			&= \frac{4c_0r^2+c_1\sqrt{r^2-c_0}-8c_0^2}{r^2}
		\end{align}
	\end{subequations}
	Now with the same line of argument we use the coordinate transformation
	\begin{align}
		r' = \frac{1}{2\sqrt{c_0}}\int D(r)dr = \sqrt{r^2-c_0}, \quad dr' = \frac{1}{2\sqrt{c_0}}D(r)dr.
	\end{align}
	such that
	\begin{align}
		B^2(r') = 4c_0\frac{r'^2+\frac{c_1}{4c_0}r'-c_0}{r'^2+c_0}.
	\end{align}
	Then defining \(l^2 \coloneqq c_0 > 0, \, m \coloneqq -\frac{c_1}{8c_0}\) we again obtain the generalized Taub-NUT metric
	\begin{subequations}
		\begin{align}
			g &= \frac{4l^2}{B^2(r')}dr'^2-B^2(r')(d\psi+\cos\theta d\varphi)^2 + (r'^2+l^2) (d\theta^2+\sin^2\theta d\varphi^2)\\
			&= \frac{r'^2+l^2}{r'^2-2mr-l^2}dr'^2 - 4l^2\frac{r'^2-2mr'-l^2}{r'^2+l^2} (d\psi+\cos\theta d\varphi)^2 + (r'^2+l^2) \, (d\theta^2+\sin^2\theta d\varphi^2),
		\end{align}
	\end{subequations}
	however describing the NUT-regions. If the orbits are not everywhere space- or timelike, we can join them smoothly along the null hypersurfaces by an extension of the form described in the last chapter.\\
\end{proof}

Thus, with respect to the constants, \(m\), \(l\) and \(n\), we have a three parameter family of vacuum space-times admitting \(SU(2) \times U(1)\) as an isometry group, such that \(SU(2)\times U(1)\) and \(SU(2)\) both have three-dimensional non-null orbits. The Taub-NUT space-time is then the unique universal cover. As in the case of the Taub-NUT space-time the generalized space-time can be considered to be a principal fibre bundle with respect to the \(U(1)\) right action with its first chern class being the constant \(n\). Furthermore, recalling the remarks in the last section, the constant \(m\) can be considered to be the Komar mass of the space-time. In particular, in this case null infinity is the Lens space \(L(n,1)\) and the NUT parameter \(l\), being the dual-Bondi-mass with respect to the infinitesimal translation induced by the Killing vector field \(-\frac{1}{2l}\partial_\psi\), is proportional to \(n\). Moreover, being a non-trivial \(S^1\) principal fibre bundle over \(S^2\) implies that the NUT parameter is non-zero.

\section{Outlook}
\label{sec:Outlook}
Taub-NUT is a very peculiar spacetime in many respects,
not only mathematically, but also concerning its 
possible physical interpretation. Yet it is frequently 
regarded for possible applications in astrophysics and 
cosmology, thereby suggesting that it may be taken 
as an adequate model for some astrophysical object.   Geodesic motions, shadows, and lensing in  
NUT-spacetime have been investigated in detail; see,
e.g.,  
\cite{Kagramanova.EtAl:2010,Jefremov.Perlick:2016,Halla.Perlick:2020}.
The question of whether and how NUT-spacetime could
be regarded as the exterior geometry produced by some 
star made of ordinary matter, like, e.g., a perfect fluid,
is an old one with partially controversial claims,
in particular regarding the physical interpretation 
of the NUT charge. 
So far no compelling physical insight seems to exists as 
to what known properties of ordinary matter could source 
a non-zero NUT charge. Perfect-fluid solutions with 
radially pointing vorticity fields have been constructed 
for that end, but the solutions established in 
\cite{Bradley.EtAl:1999} are singular, as has been
discussed in \cite{Rana:2019}.

In view of this mismatch between hypothetical physical 
applications eventually leading to measurements of the 
NUT parameter on one hand, and the lacking of a proper 
physical understanding of what might possibly be 
a matter source (if any) of it on the other, it seems a 
viable strategy to first characterise the solution as
uniquely as possible by its symmetry properties. 
This is what we attempted and achieved in this work. 
The physical problem proper clearly remains 
open for the time being. Also, the mathematical problem of classifying the
inequivalent maximal extensions of generalised
Taub-NUT should be addressed, which we plan to do in a future publication.  
\newpage
\bibliography{bibliography}

\begin{thebibliography}{10}

\bibitem{Birkhoff:RMP1923}
George~David Birkhoff.
\newblock {\em Relativity and Modern Physics}.
\newblock Harvard University Press, Harvard, Massachusetts, 1923.

\bibitem{Bradley.EtAl:1999}
Michael Bradley, Gyula Fodor, László Gergely, Mattias Marklund, and Zoltan
  Perjés.
\newblock {Rotating perfect fluid sources of the NUT metric}.
\newblock {\em Classical and Quantum Gravity}, 16(6):667–1675, 1999.

\bibitem{Griffiths.Podolsky:2009}
Jerry~B. Griffiths and Ji{\v r}\'{\i} Podolsk{\'y}.
\newblock {\em Exact Space-Times in Einstein's General Relativity}.
\newblock Cambridge University Press, Cambridge, 2009.

\bibitem{Halla.Perlick:2020}
Mourad Halla and Volker Perlick.
\newblock {Application of the Gauss–Bonnet theorem to lensing in the NUT
  metric}.
\newblock {\em General Relativity and Gravitation}, 52:112 (1--19), 2020.

\bibitem{Hong.EtAl:2012}
Sungnok Hong, John Kalliongis, Darryl McCullough, and J.~Hyam Rubinstein.
\newblock {\em Diffeomorphisms of Elliptic 3-Manifolds}.
\newblock Springer Verlag, Berlin, 2012.

\bibitem{Jebsen:1921}
J{\o}rg~Tofte Jebsen.
\newblock {{\"U}ber die allgemeinen kugelsymmetrischen L{\"o}sungen der
  Einsteinschen Gravitationsgleichungen im Vakuum}.
\newblock {\em Arkiv f{\"o}r Matematik, Astronomi och Fysik}, 15(18):1--9,
  1921.

\bibitem{Jebsen:2006}
J{\o}rg~Tofte Jebsen.
\newblock On the general spherically symmetric solutions of einstein’s
  gravitational equations in vacuo.
\newblock {\em General Relativity and Gravitation}, 37(12):2253--2259, 2006.
\newblock (English translation and reprint as `Golden Oldie' of
  \cite{Jebsen:1921}).

\bibitem{Jefremov.Perlick:2016}
Paul Jefremov and Volker Perlick.
\newblock {Circular motion in NUT space-time}.
\newblock {\em Classical and Quantum Gravity}, 33(24):245014 (1--24), 2016.
\newblock Corrigendum: Class. Quantum Grav. 35 (2018) 179501 (2pp).

\bibitem{Johansen.Ravndal:2006}
Nils~Voje Johansen and Finn Ravndal.
\newblock {On the discovery of Birkhoff’s theorem}.
\newblock {\em General Relativity and Gravitation}, 38(3):537–540, 2006.

\bibitem{Kagramanova.EtAl:2010}
Valeria Kagramanova, Jutta Kunz, Eva Hackmann, and Claus L{\"a}mmerzahl.
\newblock {Analytic treatment of complete and incomplete geodesics in Taub-NUT
  space-times}.
\newblock {\em Physical Review D}, 81(12):124044 (1--17), 2010.

\bibitem{Misner:1967}
Charles Misner.
\newblock {Taub-NUT} as a counterexample to almost anything.
\newblock In J\"urgen Ehlers, editor, {\em Relativity and Astrophysics},
  volume~8 of {\em Lectures in Applied Mathematics}, pages 160--169. American
  Mathematical Society, Providence, Rhode Island, 167.

\bibitem{Misner:1963}
Charles Misner.
\newblock The flatter regions of {Newman, Unti, and Tamburino's} generalized
  {Schwarzschild} space.
\newblock {\em Journal of Mathematical Physics}, 4(7):924--937, 1963.

\bibitem{Misner.Taub:1969}
Charles~W. Misner and Abraham~H. Taub.
\newblock A singularity-free empty universe.
\newblock {\em Soviet Physics JETP}, 28(1):122--133, 1969.

\bibitem{Moncrief:1984}
Vincent Moncrief.
\newblock {The space of (generalized) Taub-NUT spacetimes}.
\newblock {\em Journal of Geometry and Physics}, 1(1):107--130, 1984.

\bibitem{Newman.Tamburino.Unti:1963}
Ezra~Theodore Newman, Lois~A. Tamburino, and Theodore~W.J. Unti.
\newblock Empty-space generalization of the {Schwarzschild} metric.
\newblock {\em Journal of Mathematical Physics}, 4(7):915--923, 1963.

\bibitem{Ramaswamy.Sen:1981}
Sriram Ramaswamy and Amitabha Sen.
\newblock Dual‐mass in general relativity.
\newblock {\em Journal of Mathematical Physics}, 22(11):2612--2619, 1981.

\bibitem{Rana:2019}
Wajahat Rana.
\newblock {Interpretation von Sternmodellen aus idealen Fl\"ussigkeiten mit
  NUT-Ladung}.
\newblock Master's thesis, Leibniz University of Hannover, 6 2019.

\bibitem{Straumann:GR2013}
Norbert Straumann.
\newblock {\em General Relativity}.
\newblock Graduate Texts in Physics. Springer Verlag, Dordrecht, 2 edition,
  2013.

\bibitem{Taub:1951}
Abraham~Haskel Taub.
\newblock Empty space-times admitting a three parameter group of motions.
\newblock {\em Annals of Mathematics}, 53(3):472--490, 1951.

\end{thebibliography}
\bibliographystyle{plain}
\end{document}